\documentclass[a4paper]{article}

\usepackage{amsmath}
\usepackage{amssymb}
\usepackage{graphicx}
\usepackage[OT4]{fontenc}
\usepackage[latin2]{inputenc}
\usepackage{url}
\usepackage[numbers]{natbib}
\usepackage{color}
\usepackage[unicode,debug,colorlinks=true,linkcolor=black,citecolor=black,urlcolor=black,pdfusetitle=true,breaklinks=true]{hyperref}

\newcommand{\abs}[1]{\left| #1 \right|}
\newcommand{\okra}[1]{\left( #1 \right)}
\newcommand{\kwad}[1]{\left[ #1 \right]}
\newcommand{\klam}[1]{\left\{ #1 \right\}}

\newcommand{\floor}[1]{\left\lfloor #1 \right\rfloor}

\newcommand{\boole}[1]{{\bf 1}{\klam{#1}}}

\DeclareMathOperator{\sred}{\mathbf{E}_\mathit{P}}
\DeclareMathOperator*{\argmax}{arg\, max}

\DeclareMathOperator*{\hilberg}{hilb}

\newtheorem{definition}{Definition}

\newtheorem{theorem}{Theorem}

\newenvironment*{proof}{\begin{trivlist}\item[]
    \noindent\textbf{Proof:}}{$\Box$\par\end{trivlist}}
\newenvironment*{proof*}[1]{\begin{trivlist}\item[]
    \noindent\textbf{Proof of #1:}}{$\Box$\par\end{trivlist}}

\begin{document}

\begin{titlepage}

  \title{Bounds for Algorithmic Mutual Information \\ and a Unifilar
    Order Estimator}

  \author{{\L}ukasz
    D\k{e}bowski\thanks{
      {\L}. D\k{e}bowski is with the Institute of Computer Science,
      Polish Academy of Sciences, ul. Jana Kazimierza 5, 01-248
      Warszawa, Poland (e-mail: ldebowsk@ipipan.waw.pl).}}  \date{}

  \maketitle

  \begin{abstract}
    Inspired by Hilberg's hypothesis, which states that mutual
    information between blocks for natural language grows like a power
    law, we seek for links between power-law growth rate of
    algorithmic mutual information and of some estimator of the
    unifilar order, i.e., the number of hidden states in the
    generating stationary ergodic source in its minimal unifilar
    hidden Markov representation.  We consider an order estimator
    which returns the smallest order for which the maximum likelihood
    is larger than a weakly penalized universal probability. This
    order estimator is intractable and follows the ideas by Merhav,
    Gutman, and Ziv (1989) and by Ziv and Merhav (1992) but in its
    exact form seems overlooked despite some nice theoretical
    properties. In particular, we can prove both strong consistency of
    this order estimator and an upper bound of algorithmic mutual
    information in terms of it. Using both results, we show that all
    (also uncomputable) sources of a finite unifilar order exhibit
    sub-power-law growth of algorithmic mutual information and of the
    unifilar order estimator.  In contrast, we also exhibit an example
    of unifilar processes of a countably infinite order, with a
    deterministic pushdown automaton and an algorithmically random
    oracle, for which the mentioned two quantities grow as a power law
    with the same exponent. We also relate our results to natural
    language research.
    \\[2ex]
  \end{abstract}
  

\end{titlepage}

\pagestyle{plain}


\section{Introduction}
\label{secIntro}

Let $K(w)$ be the prefix-free Kolmogorov complexity and let
$J(u;w):=K(u)+K(w)-K(u,w)$ be the algorithmic mutual information
for strings $u$ and $w$ \cite{Gacs74en2,Chaitin75,LiVitanyi08}
in contrast to Shannon entropy $H(X)$ and mutual information
$I(X;Y):=H(X)+H(Y)-H(X,Y)$ for random variables $X$ and $Y$. Although
Kolmogorov complexity is in general uncomputable, there is a
hypothesis stemming from the paper by Hilberg \cite{Hilberg90}, see
also \cite{EbelingNicolis91,EbelingPoschel94,BialekNemenmanTishby01b,
  CrutchfieldFeldman03}, that the algorithmic mutual information
$J(x_1^n;x_{n+1}^{2n})$ between two blocks of a text in natural
language grows roughly like a power of the block length,
\begin{align}
  \label{Hilberg}
  J(x_1^n;x_{n+1}^{2n}) \propto n^\beta, \quad \beta>0.
\end{align}
(Notation $x_j^k$ denotes string $x_jx_{j+1}...x_k$.)  When
extrapolated to infinite texts $(x_i)_{i\in\mathbb{Z}}$ and assuming
that a computable probability model for natural language exists (which
need not be obvious), Hilberg's hypothesis implies that the respective
excess entropy of the computable model, i.e., the Shannon mutual
information between the infinite past and future
\cite{CrutchfieldFeldman03}, is infinite. Consequently, natural
language could not be optimally modeled by computable finite-state
hidden Markov processes, since they have finite excess entropy by the
data-processing inequality.  Recent progress of neural language
models, such as much publicized GPT-2 \cite{RadfordOther19} and GPT-3
\cite{BrownOthers20}, corroborates this suboptimality of finite-state
language models.  As for Hilberg's hypothesis, quite suggestive upper
bounds for the power-law growth of mutual information and some partial
evidence for divergent excess entropy can be also provided by recent
large scale computational experiments
\cite{TakahiraOthers16,HahnFutrell19,BravermanOther19,KaplanOther20}.

We have been interested in Hilberg's hypothesis long before the advent
of these computational experiments and, over years, we have formulated
a mathematical explanation thereof that links abstract semantics,
ergodic decomposition, algorithmic randomness, and information theory
\cite{Debowski06,Debowski11b,Debowski18}---for the most detailed
exposition see our book \cite{Debowski21}. This paper subscribes to
this research line making further connections and solving an open
problem stated in book \cite{Debowski21}. The aims are fourfold:
\begin{itemize}
\item The main goal of the present paper is to derive an upper bound
  for algorithmic mutual information in terms of a strongly consistent
  estimator of the unifilar order, i.e., the number of hidden states
  in the generating source in its minimal unifilar hidden Markov
  representation. What is interesting, this bound holds true not only
  for sources of a finite unifilar order but also for other stationary
  ergodic sources, where the estimator diverges to infinity or the
  process distribution is uncomputable. A similar bound for
  algorithmic mutual information but in terms of a probably
  inconsistent estimator of the Markov order, i.e., the minimal length
  of a sufficient context in the generating source, was derived in
  \cite{Debowski18}. Another upper bound for mutual information in
  terms of the number of distinct non-terminals in the shortest
  grammar-based compression
  \cite{DeMarcken96,KiefferYang00,CharikarOthers05} was demonstrated
  in \cite{Debowski11b} and linked there to Herdan-Heaps' law of
  power-law growth of vocabulary for natural language
  \cite{KuraszkiewiczLukaszewicz51en, Guiraud54,Herdan64,Heaps78}.
\item The unifilar order estimator which we apply is a modification of
  estimators of the Markov order and the hidden Markov order proposed
  by Merhav, Gutman, and Ziv \cite{MerhavGutmanZiv89} and by Ziv and
  Merhav \cite{ZivMerhav92} respectively. As in
  \cite{MerhavGutmanZiv89,ZivMerhav92}, the estimator returns the
  smallest order for which the maximum likelihood is larger than a
  penalized universal probability.  But what is interesting, our
  penalty is sublinear rather than linear which results in both no
  underestimation and no overestimation, whereas the estimators of
  \cite{MerhavGutmanZiv89,ZivMerhav92} tend to underestimate.  In the
  literature of Markov order estimation
  \cite{MerhavGutmanZiv89,CsiszarShields00,
    Csiszar02,MorvaiWeiss05,PeresShields05,DaleviDubhashi05,RyabkoAstola06,
    CsiszarTalata06,Talata13,BaigorriGocalvesResende14,RyabkoAstolaMalyutov16,
    PapapetrouKugiumtzis16} this kind of sublinear penalty can be
  traced in
  \cite{PeresShields05,RyabkoAstola06,RyabkoAstolaMalyutov16}, whereas
  it seems to have been overlooked in the literature of hidden Markov
  order estimation
  \cite{Finesso90,ZivMerhav92,WeinbergerLempelZiv92,Kieffer93,
    WeinbergerFeder94,LiuNarayan94,GassiatBoucheron03,Lehericy19,
    ShaliziShaliziCrutchfield03,ZhengHuangTong19} although most of
  these papers entertain quite similar ideas and prove strong
  consistency of related estimators.  Thus, proving strong consistency
  of our unifilar order estimator is another goal of this paper. This
  turns out to be the most technically complex achievement thereof
  since we apply theory of asymptotically mean stationary channels
  \cite{GrayKieffer80,FontanaGrayKieffer81,KiefferRahe81}. Since our
  estimator requires computing exact maximum likelihood and normalized
  maximum likelihood, it is intractable and impractical. We discuss it
  only because it has the desired theoretical property of being both
  strongly consistent and bounding algorithmic mutual information from
  above. Some practical estimators of the hidden Markov order can be
  found in \cite{Lehericy19,ZhengHuangTong19}.
\item The third goal of the paper is to exhibit clear examples of
  processes for which algorithmic mutual information and the unifilar
  order estimator grow slower than a power law (first class) or at
  least as fast as a power law (second class). By the consistency of
  the unifilar order estimator and the upper bound for algorithmic
  mutual information, all (also uncomputable) sources of a finite
  unifilar order belong to the first class. In contrast, in the second
  class we can find all perigraphic processes defined in
  \cite{Debowski18}. The defining property of perigraphic processes is
  that they describe a fixed algorithmically random infinite sequence
  effectively at a power-law rate.  As a corollary of our separation,
  the classes of finite-state hidden Markov processes and of
  perigraphic processes are disjoint, which solves one of open
  problems stated in the conclusion of our book \cite{Debowski21}.
\item Looking for simple examples of perigraphic processes, we find
  them in the class of sources of a countably infinite unifilar
  order. Our examples are called Oracle processes. The Oracle
  processes, introduced in this paper, constitute an encoding of
  perigraphic Santa Fe processes introduced in
  \cite{Debowski09,Debowski11b} into a finite alphabet but much
  simpler than discussed in \cite{Debowski12} and additionally
  satisfying the condition of unifilarity. In a nutshell, Oracle
  processes repeatedly emit random binary strings $y$, which uniquely
  represent natural numbers $\phi(y)$, followed by a comma symbol and
  by the corresponding bit $z_{\phi(y)}$ read off from an
  algorithmically random sequence $(z_k)_{k\in\mathbb{N}}$.
  Realizations of these processes can be recognized by a deterministic
  pushdown automaton with sequence $(z_k)_{k\in\mathbb{N}}$ put on an
  oracle.  Oracle processes achieve an arbitrarily large rate of
  power-law growth of algorithmic mutual information. We also show
  that our unifilar order estimator grows at the same rate for these
  sources, so in this case the bounds for algorithmic mutual
  information are tight.
\end{itemize}

Showing that finite-state hidden Markov processes cannot be
perigraphic sheds another beam of light onto the old discussion of
inadequacy of finite-state models in theoretical linguistics started
by Skinner and Chomsky
\cite{Skinner57,Chomsky56,Chomsky57,Chomsky59,ChomskyMiller59}. Chomsky's
idea was that context-free syntax of natural language is incompatible
with finite-state models advocated by Skinner. In contrast, our idea
is that human language serves for efficient description of a
potentially unboundedly complex reality in a repetitive way, therefore
a reasonable ergodic language model should be perigraphic. Whereas
mathematical connections between context-free syntax and
perigraphicness are an open problem, yet not only Hilberg's hypothesis
but also Herdan-Heaps' law of power-law growth of vocabulary
\cite{KuraszkiewiczLukaszewicz51en,Guiraud54,Herdan64,Heaps78} are
corollaries of perigraphicness
\cite{Debowski11b,DeMarcken96,KiefferYang00,CharikarOthers05}.

The topics discussed in this paper have also connections with standard
topics in information theory such as universal codes and distributions
in the style of Ryabko \cite{Ryabko88en2,Ryabko08}, normalized maximum
likelihood \cite{Shtarkov87en2,Grunwald07}, and asymptotically mean
stationary sources and channels
\cite{GrayKieffer80,FontanaGrayKieffer81,KiefferRahe81}. Consequently,
there may be natural extensions of presented ideas to model selection
for other finite-parameter families \cite{Grunwald07} but we decidedly
focus on the unifilar hidden Markov family.  To be honest, the present
paper provides constructions analogical to our preprint
\cite{Debowski20c}, which discussed a strongly consistent estimator of
the Markov order and its links with mutual information. That
manuscript was rejected in peer review as not novel enough and we have
abandoned it. We mention it merely since it resembles somewhat this
article and it cannot be withdrawn from ArXiv. Whereas the guiding
ideas are similar in both papers, the negative review inspired a few
simplifying generalizations, a few more technical steps, and many more
overt comments in this paper.

This paper adopts a linear progression of exposition, without
appendices. The first three further sections are short and discuss
preliminaries. Section~\ref{secEntities} sets up familiar concepts:
the unifilar hidden Markov family, respective (normalized) maximum
likelihood, and the Ryabko mixture. Section~\ref{secBounds} names
several obvious bounds for these, including a uniform bound for
statistical complexity.  Section~\ref{secUniversal} proves
universality of the respective Ryabko mixture, which seems also of a
preliminary character. In contrast, Section~\ref{secOrder} contains a
novel proof of strong consistency of the unifilar order
estimator. Section \ref{secMI} discusses lower and upper bounds for
power-law growth of mutual information. Section~\ref{secOracle}
contains the definition of Oracle processes and a few calculations for
them. Throughout the paper, $\ln x$ denotes the natural logarithm, in
contrast to the binary logarithm $\log x$.  $\sred X:=\int X dP$ is
the expectation.

\section{Basic entities}
\label{secEntities}

The goal of this section is to introduce familiar concepts such as the
unifilar hidden Markov family, respective (normalized) maximum
likelihood, and the Ryabko mixture---recast in our notation.  To
begin, we will consider a family of unifilar hidden Markov chains
where the number $k=1,2,3,...$ of hidden states is finite and the
emitted symbols belong to a fixed finite alphabet
$\mathbb{X}$. Unifilarity means that the hidden Markov chain is
deterministic in the automata sense, i.e., the next hidden state is a
fixed function of the previous hidden state and the previous emitted
symbol. That is, for a given sequence of symbols $x_1^n$ and states
$y_1^n$, we consider a family of unifilar hidden Markov distributions
\begin{align}
  \label{UnifilarHidden}
  \mathbb{P}(x_1^n,y_1^n|k,\pi,\tau,\varepsilon)
  &:=\pi(y_1)\varepsilon(x_1|y_1)
  \prod_{i=2}^{n}\boole{y_i=\tau(y_{i-1},x_{i-1})}\varepsilon(x_i|y_i),
\end{align}
where $\pi:\klam{1,..,k}\to[0,1]$ with $\sum_y \pi(y)=1$ is the
initial hidden state distribution,
$\tau:\klam{1,..,k}\times\mathbb{X}\to\klam{1,..,k}$ is the transition
table, and $\varepsilon:\mathbb{X}\times\klam{1,..,k}\to[0,1]$ with
$\sum_x \varepsilon(x|y)=1$ is the emission matrix.

Since so defined probabilities are prequential, i.e.,
\begin{align}
  \sum_{x_{n+1},y_{n+1}}\mathbb{P}(x_1^{n+1},y_1^{n+1}|k,\pi,\tau,\varepsilon)=
  \mathbb{P}(x_1^n,y_1^n|k,\pi,\tau,\varepsilon),
\end{align}
we can consider a joint stochastic process $(X_i,Y_i)_{i\in\mathbb{N}}$
distributed according to
\begin{align}
  P(X_1^n=x_1^n,Y_1^n=y_1^n)=
  \mathbb{P}(x_1^n,y_1^n|k,\pi,\tau,\varepsilon).
\end{align}
This process is stationary and extendable to a stationary process
$(X_i,Y_i)_{i\in\mathbb{Z}}$ if
\begin{align}
  \label{Stationary}
  \sum_{x_1,y_1}
  \pi(y_1)\varepsilon(x_1|y_1)\boole{y_2=\tau(y_1,x_1)}=\pi(y_2).
\end{align}
In any case, we define the marginal distribution
\begin{align}
  \mathbb{P}(x_1^n|k,\pi,\tau,\varepsilon):=
  \sum_{y_1^n} \mathbb{P}(x_1^n,y_1^n|k,\pi,\tau,\varepsilon)
\end{align}
and the conditional distribution
\begin{align}
  \mathbb{P}(x_1^n|k,y_1,\tau,\varepsilon):=
  \frac{1}{\pi(y_1)}\sum_{y_2^n}
  \mathbb{P}(x_1^n,y_1^n|k,\pi,\tau,\varepsilon)
  .
\end{align}

Subsequently, we define three distributions of the shape well-known
in the minimum description length theory \cite{Grunwald07}: the
maximum likelihood (ML)
\begin{align}
  \hat{\mathbb{P}}(x_1^n|k)
  &:=\max_{y,\tau,\varepsilon}
    \mathbb{P}(x_1^n|k,y,\tau,\varepsilon),
\end{align}
the normalized maximum likelihood (NML) in the spirit of Shtarkov
\cite{Shtarkov87en2}
\begin{align}
  \mathbb{P}(x_1^n|k)
  &:=
    \frac{\hat{\mathbb{P}}(x_1^n|k)}{\sum_{z_1^n\in\mathbb{X}^n}
    \hat{\mathbb{P}}(z_1^n|k)}
    \le \hat{\mathbb{P}}(x_1^n|k),
\end{align}
and the Bayesian mixture in the spirit of Ryabko \cite{Ryabko88en2,Ryabko08}
\begin{align}
  \mathbb{P}(x_1^n)
  &:=\sum_{k=1}^\infty w_k \mathbb{P}(x_1^n|k)
    ,
    \quad
    w_k:=\frac{1}{k}-\frac{1}{k+1}.
\end{align}

We notice that the maximum likelihood $\hat{\mathbb{P}}(x_1^n|k)$, the
normalized maximum likelihood $\mathbb{P}(x_1^n|k)$, and the Ryabko
mixture $\mathbb{P}(x_1^n)$ are not prequential. Moreover, the maximum
likelihood satisfies $\hat{\mathbb{P}}(x_1^n|k)=1$ for $k\ge n$, since
having as many hidden states as the string length we can put
$\pi(1)=1$, $\tau(i,x_i)=i+1$, and
$\varepsilon(x_i|i)=1$. Consequently, the normalized maximum
likelihood equals $\mathbb{P}(x_1^n|k)=\abs{\mathbb{X}}^{-n}$ for
$k\ge n$ and the Ryabko mixture $\mathbb{P}(x_1^n)$ is a computable
function of $x_1^n$ since the defining infinite series can be
truncated.  We stress that the maximum likelihood, the normalized
maximum likelihood, and the Ryabko mixture are computable in the sense
of computability theory, which will suffice for our needs of bounding
algorithmic mutual information, but they are computationally
intractable since we need to perform exhaustive search over all
transition tables $\tau$ combined with summation over exponentially
growing domains $\mathbb{X}^n$.

\section{Simple bounds}
\label{secBounds}

In this section, we will continue preliminaries and we will discuss
some simple inequalities satisfied by the maximum likelihood, the
normalized likelihood, and the Ryabko mixture. The proofs of these
facts can be easily reconstructed or found in book \cite{Grunwald07}.
Let us introduce the maximizer
\begin{align}
    \mathbb{G}(x_1^n)&:=\argmax_{k\ge 1}\mathbb{P}(x_1^n|k).
\end{align}
As explained in \cite{Grunwald07}, we observe the sandwich bound for
the Ryabko mixture 
\begin{align}
  \label{SandwichRyabko}
  -\log \mathbb{P}(x_1^n|\mathbb{G}(x_1^n))
  &\le -\log \mathbb{P}(x_1^n)
  \le -\log \mathbb{P}(x_1^n|k)-\log w_k.
\end{align}
Moreover, it is easy to see that the maximum log-likelihood is
subadditive,
\begin{align}
  \label{SubadditiveML}
  -\log \hat{\mathbb{P}}(x_1^n|k)
  -\log \hat{\mathbb{P}}(x_{n+1}^{n+m}|k)
  + \log \hat{\mathbb{P}}(x_1^{n+m}|k)
  &\le 0.
\end{align}
The above property probably was not discussed in \cite{Grunwald07}. It
follows by the particular form of family (\ref{UnifilarHidden}) and
will be used for the upper bound of algorithmic mutual information in
Section~\ref{secMI}.

Subsequently, like in \cite{Grunwald07}, we introduce the statistical
complexity of the unifilar hidden Markov family
\begin{align}
  \mathbb{C}(n|k):=
  -\log \mathbb{P}(x_1^n|k)
  +\log \hat{\mathbb{P}}(x_1^n|k)=
  \log{\sum_{z_1^n\in\mathbb{X}^n}
  \hat{\mathbb{P}}(z_1^n|k)}\le
  n\log\abs{\mathbb{X}}.
\end{align}
It is a rule of thumb that the statistical complexity of a
distribution family with exactly $k$ real parameters is roughly
$k\log n$ and there exist more exact expressions assuming some
particular conditions \cite{Grunwald07}. Here we only need a very
rough bound for $\mathbb{C}(n|k)$ but assuming that we have not only a
real-parameter emission matrix $\varepsilon$ but also an
integer-parameter transition table $\tau$. Let us denote the set of
distinct maximum likelihood parameters
\begin{align}
  \mathcal{P}_{kn}
  &:=\klam{(y,\tau,\varepsilon):
    \exists_{x_1^n} \mathbb{P}(x_1^n|k,y,\tau,\varepsilon)
    =\hat{\mathbb{P}}(x_1^n|k)}
    .
\end{align}
As explained in \cite{Grunwald07}, we can bound
\begin{align}
  \mathbb{C}(n|k)
  &=
    \log\okra{
    \sum_{x_1^n\in\mathbb{X}^n} \hat{\mathbb{P}}(x_1^n|k)
    }
    \le \log \abs{\mathcal{P}_{kn}}
    .
\end{align}
Given fixed $(y,\tau)$, likelihood
$\mathbb{P}(x_1^n|k,y,\tau,\varepsilon)$ is maximized for the
empirical distribution
\begin{align}
  \varepsilon(b|a)=
  \frac{\sum_{i=1}^n\boole{(y_i,x_i)=(a,b)}}{\sum_{i=1}^n\boole{y_i=a}}
  ,
\end{align}
where $y_1=y$ and $y_i=\tau(y_{i-1},x_{i-1})$ for $i\ge 2$.  Since
there are $k\cdot k^{k\abs{\mathbb{X}}}$ possible values of pairs
$(y,\tau)$ and given $(y,\tau)$ there are less than
$(n+1)^{k\abs{\mathbb{X}}}$ distinct empirical distributions
$\varepsilon$, we can bound the statistical complexity of the unifilar
hidden Markov family as
\begin{align}
  \label{ComplexityBound}
  \mathbb{C}(n|k)  \le \log \abs{\mathcal{P}_{kn}}
  \le
  \log \okra{
  k^{k\abs{\mathbb{X}}+1}
  (n+1)^{k\abs{\mathbb{X}}}
  }
  \le
  [k\abs{\mathbb{X}}+1]\log [k(n+1)]
  .
\end{align}
We can observe a small correction up to the mentioned rule of thumb.

\section{Universality}
\label{secUniversal}

The last relatively preliminary fact that we will present is
universality of the Ryabko mixture.  For distribution families that
contain Markov chain distributions of all orders and whose statistical
complexity of each order $k$ grows sublinearly with the sample size
$n$, the Ryabko mixture is a universal distribution by a reasoning
following the ideas of papers \cite{Ryabko88en2,Ryabko08}. It turns
out that this is the case for the unifilar hidden Markov family. For
completeness, we will present the full reasoning but we do not claim
originality of idea.  Also, as we have mentioned, this particular
Ryabko mixture is computable in the sense of computability theory but
it is intractable and highly impractical as a universal compression
procedure. We need it only for further theoretical applications.

Let $H(X):=\sred\kwad{-\log P(X)}$ be the entropy of
random variable $X$ and let $H(X|Y):=\sred\kwad{-\log P(X|Y)}$ be the
conditional entropy of $X$ given random variable $Y$.  Let
$(X_i)_{i\in\mathbb{Z}}$ be a stationary ergodic process over alphabet
$\mathbb{X}$.  We denote the conditional entropies
\begin{align}
  h_k^P:=H(X_0|X_{-k}^{-1})=\sred\kwad{-\log P(X_0|X_{-k}^{-1})}
\end{align}
and the entropy rate
\begin{align}
  h^P&:=\lim_{n\to\infty} \frac{H(X_1^n)}{n}=\inf_{k\ge 1} h_k^P=H(X_0|X_{-\infty}^{-1}).
\end{align}
The following theorem states universality of the Ryabko mixture.
\begin{theorem}
  \label{theoUniversal}
  For a stationary ergodic process $(X_i)_{i\in\mathbb{Z}}$ over
  alphabet $\mathbb{X}$, 
  \begin{align}
    \label{UniversalRyabko}
    \lim_{n\to\infty} \frac{1}{n}\kwad{-\log
    \mathbb{P}(X_1^n)}
    &= h^P \text{ $P$-a.s.},
    \\
    \label{UniversalRyabkoExp}
    \lim_{n\to\infty} \frac{1}{n}\sred\kwad{-\log
    \mathbb{P}(X_1^n)}
    &= h^P.
\end{align}
\end{theorem}
\begin{proof}
  Letting $\tau(x_1^k,x_{k+1})=x_2^{k+1}$ and
  $\varepsilon(x_{k+1}|x_1^k)=P(X_{k+1}=x_{k+1}|X_1^k=x_1^k)$, we
  obtain the conditional probability bound
  \begin{align}
    \hat{\mathbb{P}}(X_1^n|\abs{\mathbb{X}}^k)\ge
    \mathbb{P}(X_1^n|\abs{\mathbb{X}}^k,X_{-k+1}^0,\tau,\varepsilon)=
    \prod_{i=1}^n P(X_i|X_{i-k}^{i-1})
    .
  \end{align}
  Hence by the upper bound (\ref{ComplexityBound}) for the statistical
  complexity and the Birkhoff ergodic theorem, we obtain
  \begin{align}
    & \limsup_{n\to\infty} \frac{1}{n}\kwad{-\log
      \mathbb{P}(X_1^n|\abs{\mathbb{X}}^k)}
      \nonumber\\
    &\qquad =
      \limsup_{n\to\infty}
      \frac{1}{n}\kwad{-\log
      \hat{\mathbb{P}}(X_1^n|\abs{\mathbb{X}}^k)+\mathbb{C}(n|\abs{\mathbb{X}}^k)}
      \le h_k^P \text{ $P$-a.s. }
  \end{align}
  Thus by the upper bound in (\ref{SandwichRyabko}) and the Barron
  lemma \cite[Theorem 3.1]{Barron85b}, we obtain
  (\ref{UniversalRyabko}). Noticing that
  $\mathbb{P}(X_1^n)\ge w_n
  \mathbb{P}(X_1^n|n)=w_n\abs{\mathbb{X}}^{-n}$, we hence obtain
  (\ref{UniversalRyabkoExp}) by dominated convergence.
\end{proof}

\section{Order estimation}
\label{secOrder}

We are moving on to more difficult reasonings.  In this section, we
will construct a simple unifilar order estimator and we will prove
that it is strongly consistent and asymptotically unbiased.  The
estimator is intractable but it leads to simple bounds for algorithmic
mutual information to be discussed in Section~\ref{secMI}.  To prove
its strong consistency we will apply universality of the Ryabko
mixture demonstrated in Section~\ref{secUniversal}.

Let $(X_i)_{i\in\mathbb{Z}}$ be a stationary ergodic process over
alphabet $\mathbb{X}$ as before. We define the unifilar order of the
process as
\begin{align}
 M^P
  &:=\inf\klam{k:\exists_{\pi,\tau,\varepsilon}\forall_{n\ge 1,x_1^n}
    P(X_1^n=x_1^n)=
    \mathbb{P}(x_1^n|k,\pi,\tau,\varepsilon)}
\end{align}
with the convention that the infimum of the empty set is infinite.
Subsequently, we will consider a unifilar order estimator which is a
certain modification of estimators of the Markov order and the hidden
Markov order proposed by Merhav, Gutman, and Ziv
\cite{MerhavGutmanZiv89} and by Ziv and Merhav \cite{ZivMerhav92}
respectively. The idea of \cite{MerhavGutmanZiv89,ZivMerhav92} is that
the estimator returns the smallest order for which the maximum
likelihood is larger than a penalized universal probability.
Consequently, we will define the unifilar order estimator
\begin{align}
  \label{Estimator}
  \mathbb{M}(x_1^n)
  &:=\min\klam{k:
    \hat{\mathbb{P}}(x_1^n|k)\ge
    w_n\mathbb{P}(x_1^n)}
    ,
    \quad
    w_n:=\frac{1}{n}-\frac{1}{n+1}
    .
\end{align}
We can see that the estimator is nicely bounded by
$\mathbb{M}(x_1^n)\le n$ since $\hat{\mathbb{P}}(x_1^n|k)=1$ for
$k\ge n$. In contrast to the hidden Markov order estimator by Ziv and
Merhav, we use much smaller penalty $-\log w_n=o(n)$, whereas penalty
$-\log w_n=\lambda n$ was used in \cite{ZivMerhav92}. As we will see,
our choice of $w_n$ results in strong consistency whereas the
estimator by \cite{ZivMerhav92} suffered from underestimation.

In the literature of Markov order estimation
\cite{MerhavGutmanZiv89,CsiszarShields00,
  Csiszar02,MorvaiWeiss05,PeresShields05,DaleviDubhashi05,RyabkoAstola06,
  CsiszarTalata06,Talata13,BaigorriGocalvesResende14,RyabkoAstolaMalyutov16,
  PapapetrouKugiumtzis16}, a kind of sublinear penalty
$-\log w_n=o(n)$ in estimators resembling (\ref{Estimator}) can be
traced in \cite{PeresShields05,RyabkoAstola06,RyabkoAstolaMalyutov16},
whereas it seems overlooked in variants of estimator (\ref{Estimator})
in the literature of hidden Markov order estimation
\cite{Finesso90,ZivMerhav92,WeinbergerLempelZiv92,Kieffer93,
  WeinbergerFeder94,LiuNarayan94,GassiatBoucheron03,Lehericy19,
  ShaliziShaliziCrutchfield03,ZhengHuangTong19} although majority of
these papers treat quite similar ideas and prove strong consistency of
related estimators.  Penalty $-\log w_n=o(n)$ was also considered in
the rejected paper \cite{Debowski20c} for Markov order estimation. In
this case, minus logarithm of maximum likelihood equals the empirical
conditional entropy and can be quickly computed.

In contrast, the unifilar order estimator (\ref{Estimator}) is
computable in the sense of computability theory but it is intractable
since it applies exact maximum likelihood and normalized maximum
likelihood. We will need it as is for the next section since it yields
the most elegant upper bound for the algorithmic mutual
information. Ignoring the question of obtaining this bound for a
while, we note that we can make the estimator somewhat computationally
simpler while preserving strong consistency if we replace universal
distribution $\mathbb{P}(x_1^n)$ with a simpler universal compression
procedure such as the Lempel-Ziv code \cite{ZivLempel77}. This idea
was proposed by Merhav, Gutman, and Ziv \cite{MerhavGutmanZiv89} and
by Ziv and Merhav \cite{ZivMerhav92} themselves. This substitution,
however, breaks the simple upper bound for mutual information while
not solving the problem of computing the exact maximum likelihood,
which requires combinatorial optimization over all transition tables
$\tau$.  In contrast, some practical estimators of the hidden Markov
order can be found in \cite{Lehericy19,ZhengHuangTong19}.

The following theorem states strong consistency of the unifilar order
estimator. The proof technique for the impossibility of overestimation
can be considered to be transferred from Markov order estimation proof
ideas such as \cite{RyabkoAstola06,RyabkoAstolaMalyutov16}. We suppose
that the proof of the impossibility of underestimation is more
original. Since we apply a result about asymptotically mean stationary
channels by Kieffer and Rahe \cite{KiefferRahe81}, we suspected that
Kieffer \cite{Kieffer93} might have used a similar technique in the
context of hidden Markov order estimation but we did not find it
there.
\begin{theorem}
  \label{theoConsistent}
  For a stationary ergodic process $(X_i)_{i\in\mathbb{Z}}$ over
  alphabet $\mathbb{X}$,
 \begin{align}
   \lim_{n\to\infty} \mathbb{M}(X_1^n) &=M^P \text{ $P$-a.s.}
 \end{align}
 and we have the overestimation bound
 $P\okra{\mathbb{M}(X_1^n)>M^P}\le w_n$.
\end{theorem}
\begin{proof}
  Our proof is split into impossibility of overestimation and of
  underestimation.  The bound for the overestimation probability is
  received by inequality
  $\hat{\mathbb{P}}(X_1^n|M^P)\ge P(X_1^n|Y_1)$, where $Y_1$
  is the hidden state emitting $X_1$, and by the Barron lemma
  \cite[Theorem 3.1]{Barron85b}. Hence
  \begin{align}
    P\okra{\mathbb{M}(X_1^n)>M^P}
    &\le
      P\okra{\hat{\mathbb{P}}(X_1^n|M^P)< w_n\mathbb{P}(X_1^n)}
      \nonumber\\
    &\le
      P\okra{\frac{w_n\mathbb{P}(X_1^n)}{P(X_1^n|Y_1)}>1}
      \le
      w_n.
  \end{align}
  Since $\sum_{n=1}^\infty w_n=1$, the impossibility of overestimation
  follows by the Borel-Cantelli lemma.

  Now, we demonstrate the impossibility of underestimation, which is
  much more difficult to see. Since the Ryabko mixture was shown to be
  universal in the sense of (\ref{UniversalRyabko}) and the penalty is
  $-\log w_n=o(n)$, it is sufficient to show that
  \begin{align}
    \liminf_{n\to\infty} \frac{1}{n}\kwad{-\log
    \hat{\mathbb{P}}(X_1^n|k)}
    &
      > h^P \text{ $P$-a.s. for } k<M^P.
  \end{align}  
  Our reasoning will go by showing that the left hand side of the
  above inequality equals almost surely a sort of conditional entropy
  $h_{[k]}^P$ which is strictly greater than $h^P$ if $k<M^P$.
  
  We observe first that for any finite set $\mathcal{M}$,
  \begin{align}
    \liminf_{n\to\infty} \min_{m\in\mathcal{M}} a_{nm}
    =
    \min_{m\in\mathcal{M}} \liminf_{n\to\infty} a_{nm}.
  \end{align}
  It is so since we can take sufficiently large $n$ on both sides to
  interchange the infimums. In our case, the set of pairs $(y,\tau)$
  for a fixed $k$ is finite.  Hence
  \begin{align}
    \liminf_{n\to\infty} \min_{y,\tau,\varepsilon}
    \frac{1}{n}\kwad{-\log
    \mathbb{P}(X_1^n|k,y,\tau,\varepsilon)}
    =
    \min_{y,\tau}
    \liminf_{n\to\infty} \min_{\varepsilon}
    \frac{1}{n}\kwad{-\log
    \mathbb{P}(X_1^n|k,y,\tau,\varepsilon)}.
  \end{align}


  Now we will apply technically quite a difficult but beautiful result
  by Kieffer and Rahe \cite{KiefferRahe81}, which says that an ergodic
  Markov channel applied to an ergodic asymptotically mean stationary
  process yields a jointly ergodic asymptotically mean stationary
  process. Denote $Y^{y,\tau}_1:=y\in\klam{1,...,k}$ and
  $Y^{y,\tau}_{i+1}:=\tau(Y^{y,\tau}_i,X_i)$. We can see that the
  distribution of $(Y^{y,\tau}_i)_{i\in\mathbb{N}}$ given
  $(X_i)_{i\in\mathbb{N}}$ is an ergodic Markov channel, whereas
  process $(X_i)_{i\in\mathbb{N}}$ is stationary ergodic. Thus process
  $(X_i,Y^{y,\tau}_i)_{i\in\mathbb{N}}$ is asymptotically mean
  stationary ergodic. Let process
  $(\bar X_i,\bar Y^{y,\tau}_i)_{i\in\mathbb{Z}}$ be distributed
  according to the stationary mean of
  $(X_i,Y^{y,\tau}_i)_{i\in\mathbb{N}}$.  Since
  $(X_i)_{i\in\mathbb{N}}$ is stationary, we can assume without loss
  of generality that
  $(\bar X_i)_{i\in\mathbb{N}}=(X_i)_{i\in\mathbb{N}}$.  Moreover, by
  definition of the stationary mean, recursion
  $\bar Y_{i+1}^{y,\tau}=\tau(\bar Y^{y,\tau}_i,X_i)$ holds by
  recursion $Y^{y,\tau}_{i+1}=\tau(Y^{y,\tau}_i,X_i)$.  (Notice,
  however, that we cannot assume
  $\bar Y^{y,\tau}_i=\sigma_{y,\tau}(X_{-\infty}^{i-1})$ since there
  is a simple counterexample: a periodic process
  $(Y^{y,\tau}_i)_{i\in\mathbb{N}}$ with a constant process
  $(X_i)_{i\in\mathbb{N}}$.)

  The beauty of asymptotically mean stationary processes lies in the
  fact that we have a generalization of the Birkhoff ergodic theorem
  \cite{GrayKieffer80}. The claim is that the Ces\`aro averages
  converge almost surely to expectations with respect to the
  stationary mean. Hence, by the application of the Birkhoff ergodic
  theorem to empirical counts in the most likely distribution
  $\varepsilon$ given $(y,\tau)$, we obtain
  \begin{align}
    &
      \min_{y,\tau}
      \liminf_{n\to\infty} \min_{\varepsilon}
      \frac{1}{n}\kwad{-\log
      \mathbb{P}(X_1^n|k,y,\tau,\varepsilon)}
      \nonumber\\
    &\quad =
      \min_{y,\tau} \liminf_{n\to\infty} \min_{\varepsilon}
      \frac{1}{n}\sum_{i=1}^n\kwad{-\log \varepsilon(X_i|Y^{y,\tau}_i)}
      \nonumber\\
    &\quad =
      h_{[k]}^P:=\min_{y,\tau}
      \sred
      \kwad{-\log P(X_i|\bar Y^{y,\tau}_i)} \text{ $P$-a.s.}
  \end{align}

  Let $y$ and $\tau$ be some minimizing parameters and let us
  abbreviate $\bar Y_i:=\bar Y^{y,\tau}_i$. What happens if
  $h_{[k]}^P:=H(X_i|\bar Y_i)=h^P$? Since
  $\bar Y_{i+1}=\tau(\bar Y_i,X_i)$, we can write
  \begin{align}
    H(X_i|\bar Y_i)-H(X_i|X_1^{i-1},\bar Y_1)
    &=
      I(X_i;X_1^{i-1},\bar Y_1^{i-1}|\bar Y_i).
  \end{align}
  Hence by stationarity of $(X_i,\bar Y^{y,\tau}_i)_{i\in\mathbb{Z}}$, 
  \begin{align}
    \sum_{i=1}^n H(X_i|\bar Y_i)
    = H(X_1^n|\bar Y_1)+
    \sum_{i=0}^{n-1} I(X_j;X_{j-i}^{j-1},\bar Y_{j-i}^{j-1}|\bar Y_j).
  \end{align}
  Dividing by $n$ and letting $n\to\infty$ yields
  \begin{align}
      h_{[k]}^P:=H(X_i|\bar Y_i)= h^P+
    I(X_j;X_{-\infty}^{j-1},\bar Y_{-\infty}^{j-1}|\bar Y_j),
  \end{align}
  where we freely apply Shannon information measures for arbitrary,
  also infinite $\sigma$-fields, whose properties were described in
  \cite{Wyner78,Debowski20}.  That is, if $h_{[k]}^P=h^P$ then
  $I(X_i;X_{-\infty}^{i-1},\bar Y_{-\infty}^{j-1}|\bar Y_i)=0$. Since also
  $\bar Y_{i+1}=\tau(\bar Y_i,X_i)$ then $(X_i)_{i\in\mathbb{N}}$ is a
  unifilar hidden Markov process with $\le k$ hidden states
  distributed according to $\bar Y_i$. Consequently, we have
  $M^P\le k$.
\end{proof}

The above result implies that the estimator is asymptotically
unbiased, which we will need in the next section.
\begin{theorem}
  \label{theoUnbiased}
  For a stationary ergodic process $(X_i)_{i\in\mathbb{Z}}$ over
  alphabet $\mathbb{X}$,
  \begin{align}
    \lim_{n\to\infty} \sred \mathbb{M}(X_1^n) = M^P.
  \end{align}
\end{theorem}
\begin{proof}
  By $\mathbb{M}(x_1^n)\le n$ and by the overestimation bound in
  Theorem \ref{theoConsistent} we have
  \begin{align}
    \sred \mathbb{M}(X_1^n)\le M^P +
    nP(\mathbb{M}(x_1^n)>M^P)= M^P+\frac{1}{n+1}.
  \end{align}
  On the other hand, by the Fatou lemma,
  \begin{align}
    M^P=
    \sred \liminf_{n\to\infty} \mathbb{M}(X_1^n)\le
    \liminf_{n\to\infty} \sred \mathbb{M}(X_1^n).
  \end{align}
  Hence the claim follows.
\end{proof}


\section{Mutual information}
\label{secMI}

In this section, we will present the culmination of this paper,
namely, various bounds for the asymptotic power-law growth of
algorithmic mutual information and related quantities such as the
hidden Markov order estimator introduced in the previous section.  For
this aim we will apply the basic bounds presented in Section
\ref{secBounds}, universality of the Ryabko mixture from Section
\ref{secUniversal}, and strong consistency of the hidden Markov order
estimator from Section~\ref{secOrder}.

For the sake of further considerations concerning the power-law growth
of various quantities, let us introduce so called
Hilberg exponents
\begin{align}
  \hilberg_{n\to\infty} s(n):=
  \limsup_{n\to\infty}
  \frac{\log \max\klam{1,s(n)}}{\log n}
\end{align}
for real functions $s(n)$ of natural numbers, cf.\
\cite{Debowski15d,Debowski18,Debowski21}, where we gradually
approached the above definition.  The Hilberg exponents capture the
asymptotic power-law growth of the respective functions, such as
$\hilberg_{n\to\infty} n^\beta=\beta$ for $\beta\ge 0$.  We will begin
with a strengthening of a simple observation from
\cite{Debowski18,Debowski21}. Our improvement is also very simple and
it consists in replacing condition $\mathfrak{J}(n)\ge -C$ with
$\mathfrak{S}(n)-n\mathfrak{s}\ge -C$ as sufficient for equality of
the respective Hilberg exponents. It is curious that we have not
noticed this earlier.
\begin{theorem}[cf.\ \cite{Debowski18,Debowski21}]
  \label{theoHilbergRedundancy}
  For a function $\mathfrak{S}:\mathbb{N}\rightarrow\mathbb{R}$,
  define $\mathfrak{J}(n):=2\mathfrak{S}(n)-\mathfrak{S}(2n)$. If
  $\lim_{n\rightarrow\infty} \mathfrak{S}(n)/n=\mathfrak{s}$ for a
  $\mathfrak{s}\in\mathbb{R}$ then
\begin{align}
  \label{HilbergRedundancy}
  \hilberg_{n\rightarrow\infty} \okra{\mathfrak{S}(n)-n\mathfrak{s}}
  \le
  \hilberg_{n\rightarrow\infty} \mathfrak{J}(n)
\end{align}
with an equality if $\mathfrak{S}(n)-n\mathfrak{s}\ge -C$ for all but
finitely many $n$ and some $C>0$.
\end{theorem}
\begin{proof}
  Write $\delta=\hilberg_{n\rightarrow\infty} \mathfrak{J}(n)$.
  The proof of $\hilberg_{n\rightarrow\infty}
  \okra{\mathfrak{S}(n)-n\mathfrak{s}}\le \delta$ can be found
  in \cite{Debowski18,Debowski21}.
  Now assume that $\mathfrak{S}(n)-n\mathfrak{s}\ge -C$ for all but
  finitely many $n$. We have then
  \begin{align}
    \mathfrak{S}(n)-n\mathfrak{s}=
    \frac{\mathfrak{J}(n)}{2}+\frac{\mathfrak{S}(2n)-2n\mathfrak{s}}{2}
    \ge \frac{\mathfrak{J}(n)-C}{2}
  \end{align}
  for sufficiently large $n$. Hence
  $\delta\le
  \hilberg_{n\rightarrow\infty}\okra{\mathfrak{S}(n)-n\mathfrak{s}}$.
  Thus we obtain the equality in (\ref{HilbergRedundancy}).
\end{proof}

Let $(X_i)_{i\in\mathbb{Z}}$ be a stationary ergodic process over
alphabet $\mathbb{X}$. As discussed in
\cite{CrutchfieldFeldman03,Debowski21}, we can equivalently define the
excess entropy 
\begin{align}
  E^P:=\lim_{n\to\infty} \kwad{H(X_1^n)-nh^P}=
  \lim_{n\to\infty}I(X_1^n;X_{n+1}^{2n}).
\end{align}
By the data-processing inequality for the Shannon mutual information,
we obtain $E^P\le\log M^P$.  That is, the excess entropy is finite for
finite-state hidden Markov processes (also non-unifilar ones). In the
next turn, we may ask how fast the defining expressions diverge to
infinity. By Theorem \ref{theoHilbergRedundancy} and
$I(X_1^n;X_{n+1}^{2n})=2H(X_1^n)-H(X_1^{2n})$ from stationarity, we
can define the Hilberg exponent
\begin{align}
  \beta^P_1:=
  \hilberg_{n\to\infty} \kwad{H(X_1^n)-nh^P}=
  \hilberg_{n\to\infty} I(X_1^n;X_{n+1}^{2n})\le 1.
\end{align}
Let us compare this result with Kolmogorov complexity and algorithmic
mutual information.  By the theorem of Brudno \cite{Brudno82}, the
prefix-free Kolmogorov complexity is the length of a universal code,
  \begin{align}
    \label{UniversalK}
    \lim_{n\to\infty} \frac{K(X_1^n)}{n}
    &= h^P \text{ $P$-a.s.},
    \\
    \label{UniversalKExp}
    \lim_{n\to\infty} \frac{\sred K(X_1^n)}{n}
    &= h^P.
\end{align}
Since $\sred K(X_1^n)\ge H(X_1^n)$ by the prefix-free property,
Theorem \ref{theoHilbergRedundancy} yields
\begin{align}
  \beta^P_1\le \beta^P_2:=
  \hilberg_{n\to\infty} \sred \kwad{K(X_1^n)-nh^P}=
  \hilberg_{n\to\infty} \sred J(X_1^n;X_{n+1}^{2n})\le 1.
\end{align}
Similarly, by universality of the Ryabko mixture proved in Theorem
\ref{theoUniversal} and inequality
\begin{align}
  K(x_1^n)\le -\log \mathbb{P}(X_1^n)-\log w_n+K(\mathbb{P})
\end{align}
from computability of the Ryabko mixture and Shannon-Fano coding,
we obtain
\begin{align}
  \beta^P_1\le \beta^P_2\le \beta^P_3
  &:=
    \hilberg_{n\to\infty} \sred\kwad{-\log \mathbb{P}(X_1^n)-nh^P}
    \nonumber\\
  &=
    \hilberg_{n\to\infty}
    \sred\kwad{-\log \mathbb{P}(X_1^n)-\log
    \mathbb{P}(X_{n+1}^{2n})+\log \mathbb{P}(X_1^{2n})}\le 1.
    \label{HilbergCompare}
\end{align}
The above bounds besides the second equality in (\ref{HilbergCompare})
(stated as inequality in in \cite{Debowski18}) were observed in the
context of universal coding in \cite{Debowski18}.

Consequently in \cite{Debowski18}, we have bounded Hilberg exponents
$\beta^P_1$, $\beta^P_2$, and $\beta^P_3$ in terms of so called
theorems about facts and words. These mathematical propositions state
that under some formalization, the number of independent binary facts
predictable from a finite text is roughly less than the mutual
information between two halves of the text and this is roughly less
than the number of distinct words detectable in the text. Namely, the
theorems about facts and words take form of inequalities
\begin{align}
  \label{DefaultFactsWords}
  \hilberg_{n\to\infty} \sred U(X_1^n)
  \le
  \beta^P_i
  \le
  \hilberg_{n\to\infty} \sred V(X_1^n),
\end{align}
where $U(x_1^n)$ is approximately the number of independent binary
facts effectively predictable from a fixed text $x_1^n$ and $V(x_1^n)$
is approximately the number of distinct words detectable in a fixed
text $x_1^n$. On a conceptual level, one may suppose that inequalities
(\ref{DefaultFactsWords}) are echoing to a certain extent inequalities
between mutual information and the common information functions by
G\'acs and K\"orner \cite{GacsKorner73} (less than mutual
information) and by Wyner \cite{Wyner75} (larger than mutual
information). The exact formulation of the theorems about facts and
words rests on a few particular modeling assumptions and allows for
some freedom with respect to the definition of $U(x_1^n)$ and
$V(x_1^n)$, cf.\ \cite{Debowski11b,Debowski18}.

Whereas $V(x_1^n)$ can be actually taken as a number of entities that
resemble orthographic words in the ordinary sense, cf.\
\cite{Debowski11b,DeMarcken96,KiefferYang00,CharikarOthers05}, we can
look for other possibilities. Following \cite{Debowski18} and
\cite{Debowski20c}, where $V(x_1^n)$ was taken as the number of
distinct substrings of the length of a Markov order estimator
(inconsistent in \cite{Debowski18} and strongly consistent in
\cite{Debowski20c}), we can suspect that one can also take $V(x_1^n)$
to be a hidden Markov order estimator. This is exactly the case. The
following reasoning squeezes to a triviality our earlier ideas for
proving the right inequality in (\ref{DefaultFactsWords}), cf.\
\cite{Debowski11b,Debowski18}.
\begin{theorem}
  \label{theoMIOrder}
  For a stationary process $(X_i)_{i\in\mathbb{Z}}$ over
  alphabet $\mathbb{X}$,
  \begin{align}
    \beta^P_3\le \beta^P_{\mathbb{M}}:=
    \hilberg_{n\to\infty} \sred \mathbb{M}(X_1^n).
  \end{align}
\end{theorem}
\begin{proof}
Denoting $k:=\mathbb{M}(X_1^{2n})\le 2n$, we observe by
(\ref{SandwichRyabko}) and (\ref{SubadditiveML}) that
\begin{align}
  &-\log \mathbb{P}(X_1^n)-\log \mathbb{P}(X_{n+1}^{2n})
    +\log \mathbb{P}(X_1^{2n})
    \nonumber\\
  &\quad\le    
    -\log \mathbb{P}(X_1^n|k)
    -\log \mathbb{P}(X_{n+1}^{2n}|k)
    -2\log w_k
    +\log \hat{\mathbb{P}}(X_1^{2n}|k)-\log w_{2n}
     \nonumber\\
  &\quad\le
    2\mathbb{C}(n|k)-2\log w_k-\log w_{2n}
    -\log \hat{\mathbb{P}}(X_1^n|k)
    -\log \hat{\mathbb{P}}(X_{n+1}^{2n}|k)
    +\log \hat{\mathbb{P}}(X_1^{2n}|k)
     \nonumber\\
  &\quad\le
    2\mathbb{C}(n|k)-2\log w_k-\log w_{2n}.
\end{align}
Hence by (\ref{ComplexityBound}) we obtain
\begin{align}
  \beta^P_3
  \le 
  \hilberg_{n\to\infty} \sred \mathbb{C}(n|\mathbb{M}(X_1^{2n}))
  \le
  \hilberg_{n\to\infty} \sred \mathbb{M}(X_1^n)
  .
\end{align}
\end{proof}
By Theorems \ref{theoUnbiased} and \ref{theoMIOrder}, for processes of
a finite unifilar order, we have not only $\beta^P_1=0$ but also
$\beta^P_2=\beta^P_3=\beta^P_{\mathbb{M}}=0$. Equality $\beta^P_2=0$
carries over also to non-unifilar finite-state hidden Markov processes
by the data-processing inequality for algorithmic mutual information.

Consequently, we can ask whether there exist processes such that
$\beta^P_2$, $\beta^P_3$, and $\beta^P_{\mathbb{M}}$ are arbitrarily 
close to $1$. The answer is given by a particular form of inequalities
(\ref{DefaultFactsWords}). We have the following proposition.
\begin{theorem}[\cite{Debowski18}]
  \label{theoFactsMI}
  Consider an arbitrary computable function
  $g:\mathbb{N}\times\mathbb{X}^*\to\klam{0,1,2}$ and an arbitrary fixed
  algorithmically random binary sequence $z=(z_k)_{k\in\mathbb{N}}$ (in
  the Martin-L\"of sense), i.e., we have $K(z_1^n)\ge n-c$ for a
  certain constant $c$. Define
 \begin{align}
  \label{Facts}
    U_g(x_1^n|z):=\min\klam{k\ge 1: g(k,x_1^n)\neq z_k}.
  \end{align}
  For a stationary process
  $(X_i)_{i\in\mathbb{Z}}$ over alphabet $\mathbb{X}$,
  \begin{align}
    \label{HilbergFacts}
    \beta^P_{g,z}:=
    \hilberg_{n\to\infty} \sred U_g(X_1^n|z)\le \beta^P_2.
  \end{align}
\end{theorem}
The guiding intuition is behind this result is that sequence
$z=(z_k)_{k\in\mathbb{N}}$ is the abstract pool of independent binary
``facts'' $z_k$ and $U_g(x_1^n|z)-1$ is the number of initial facts
that can be correctly predicted given text $x_1^n$ using an effective
prediction procedure. From this point of view, quantity
$U_g(x_1^n|z)-1$ can be seen as an approximation of the common
information by G\'acs and K\"orner which was demonstrated to be not
greater than the mutual information in \cite{GacsKorner73}.

Let us entertain this definition.
\begin{definition}[\cite{Debowski18}]
  A stationary process $(X_i)_{i\in\mathbb{Z}}$ is called perigraphic
  if inequality $\beta^P_{g,z}>0$ holds for Hilberg exponent
  (\ref{HilbergFacts}) given some computable function
  $g:\mathbb{N}\times\mathbb{X}\to\klam{0,1,2}$ and some
  algorithmically random binary sequence $z=(z_k)_{k\in\mathbb{N}}$.
\end{definition}
By Theorems \ref{theoUnbiased}, \ref{theoMIOrder}, and
\ref{theoFactsMI}, the class of finite-state hidden Markov processes
and the class of perigraphic processes are disjoint, which solves an
important open problem stated in the conclusion of book
\cite{Debowski21}.

The motivating idea of this open problem was that human language
serves for efficient description of a potentially unboundedly complex
reality in a repetitive way. If we equate the complex reality with an
infinite algorithmically random sequence, then a reasonable
statistical language model should be perigraphic. We wanted to show
that for that reason, a reasonable language model cannot be a
finite-state hidden Markov process, which is a different cause against
finite-state models than context-free syntax postulated by Chomsky
\cite{Chomsky56,Chomsky57,Chomsky59}.

As shown in \cite{Debowski18}, perigraphic processes have uncomputable
distributions and they look as typical ergodic components of some
strongly non-ergodic processes. In the next section, however, we will
show that some perigraphic processes lie quite low in the hierarchy of
stochastic processes since there are simple perigraphic unifilar
processes with a countably infinite number of hidden states, which we
call Oracle processes. In particular, Oracle processes do not exhibit
hierarchical context-free structures of an infinite depth. For
processes that exhibit context-free hierarchical structures of an
infinite depth and are possibly perigraphic (we have not proved it),
we refer to \cite{Debowski17,Debowski21}.

\section{Oracle processes}
\label{secOracle}

In this section we will construct some simple unifilar processes,
called Oracle processes, which are perigraphic. Necessarily, these
processes have a countably infinite number of hidden states. We will
show that for these processes
$\beta^P_{g,z}=\beta^P_2=\beta^P_3=\beta^P_{\mathbb{M}}=\beta$, where
$\beta$ is an arbitrary number in range $(0,1)$. That is, the bounds
given by Theorems \ref{theoMIOrder} and \ref{theoFactsMI} for the
power-law growth of the number of predictable facts, of algorithmic
mutual information, and of the unifilar order estimator are tight also
in this case. This tightness seems quite a new result in our little
theory of perigraphic processes.

To begin, let us recall that given a fixed algorithmically random
binary sequence $z=(z_k)_{k\in\mathbb{N}}$ such as the typical result
of unbiased coin-flip or the binary expansion of Chaitin constant
$\Omega=(\Omega_k)_{k\in\mathbb{N}}$ \cite{Chaitin75}, we can
construct a very simple example of a perigraphic process, called the
Santa Fe process \cite{Debowski09,Debowski18}. The Santa Fe process
$(X_i)_{i\in\mathbb{Z}}$ is a process over a countably infinite
alphabet and it consists of pairs
\begin{align}
  \label{SantaFe}
  X_i=(K_i,z_{K_i}),
\end{align}
where $K_i$ is an IID process taking values in natural numbers
according to Zipf's law $P(K_i=k)\propto k^{-\alpha}$, $\alpha>1$.  It
can be easily shown that in this case $\beta^P_{g,z}=1/\alpha$ for
quite an obvious choice of predictor $g(k,x_1^n)$ that simply reads
off value $z_k$ from pair $(k,z_k)$ if it appears in block $x_1^n$ and
returns $2$ otherwise \cite{Debowski18,Debowski21}. A similar
perigraphic process over a finite alphabet can be constructed through
stationary variable-length coding of the Santa Fe process
\cite{Debowski10,Debowski12,Debowski21} or more directly, as a
unifilar process with a countably infinite number of hidden states to
be introduced in this paper under the name of Oracle processes.

First, we will generalize the concept of a unifilar process to a
countable set of hidden states. In particular, the set of hidden
states can be exemplified as the set of natural numbers $\mathbb{N}$
or as the set of strings $\mathbb{X}^*$.
\begin{definition}
  A possibly non-stationary process $(X_i)_{i\in\mathbb{N}}$ over an
  alphabet $\mathbb{X}$ is called unifilar with respect to a process
  $(Y_i)_{i\in\mathbb{N}}$ if $(Y_i)_{i\in\mathbb{N}}$ is a
  homogeneous first order Markov process over a countable alphabet
  $\mathbb{Y}$ such that
\begin{enumerate}
\item $P(Y_1=y_1)=\pi(y_1)$, 
\item $P(X_i=x_i|Y_1^i=y_1^i,X_1^{i-1}=x_1^{i-1})=\varepsilon(x_i|y_i)$, 
\item $Y_{i+1}=\tau(Y_i,X_i)$
\end{enumerate}
for certain functions $\pi:\mathbb{Y}\to[0,1]$,
$\varepsilon:\mathbb{X}\times\mathbb{Y}\to[0,1]$, and
$\tau:\mathbb{Y}\times\mathbb{X}\to\mathbb{Y}$.  
\end{definition}

The following Oracle($\theta$) process is the Santa Fe process in
disguise. The idea is that we first emit some random string $y2$
uniquely representing the natural number $\phi(y)$ and then we emit
the corresponding bit $z_{\phi(y)}$ read off from the oracle being an
algorithmically random sequence $z=(z_k)_{k\in\mathbb{N}}$. Once this
bit is emitted, we repeat the procedure ad infinitum.
\begin{definition}
  Let $\psi:\mathbb{N}\to\klam{0,1}^*$ where $\psi(k)$ is the binary
  expansion of number $k$ stripped of the initial digit $1$:
  $\psi(1)=\lambda$, $\psi(2)=0$, $\psi(3)=1$, $\psi(4)=00$, ...  Let
  $\phi=\psi^{-1}$ be the inverse function. Let
  $z=(z_k)_{k\in\mathbb{N}}$ be an algorithmically random binary
  sequence.  Oracle($\theta$) process with parameter $\theta\in[0,1]$
  is the unifilar process defined by:
\begin{itemize}
\item $\mathbb{X}=\klam{0,1,2}$,
\item $\mathbb{Y}=\klam{a,b}\times\klam{0,1}^*$,
\item $\varepsilon(x|ay)=\theta/2$ and $\tau(ay,x)=ayx$ for
  $x\in\klam{0,1}$ and $y\in\klam{0,1}^*$,
\item $\varepsilon(2|ay)=(1-\theta)$ and $\tau(ay,2)=by$ for
  $y\in\klam{0,1}^*$,
\item $\varepsilon(z_{\phi(y)}|by)=1$ and $\tau(by,z_{\phi(y)})=a$ for
  $y\in\klam{0,1}^*$.
\end{itemize}  
\end{definition}
As we can see, realizations of Oracle($\theta$) process are recognized
by a deterministic pushdown automaton combined with an algorithmically
random oracle: First, the random output binary string $y$ is pushed on
the stack and upon producing symbol $2$, the stack is emptied with an
expectation of symbol $z_{\phi(y)}$ as a next output. Having met this
expectation, the stack is ready to be refilled. This observation
indicates that there may be some interactions between perigraphicness
and context-free languages or syntax if we assume a finite output
alphabet. It is a matter of future research to establish more precise
connections.

Since we have not discussed the Oracle processes before, as a warm-up,
let us compute the stationary distribution and the entropy rate of an
Oracle process.  There exists a simple expression for the entropy rate
of a unifilar process.
\begin{theorem}[\cite{TraversCrutchfield11,TraversCrutchfield11b,
    TraversCrutchfield14}]
  \label{theoUnifilar}
  The entropy rate of a process $(X_i)_{i\in\mathbb{N}}$ unifilar with
  respect to a stationary process $(Y_i)_{i\in\mathbb{N}}$ with
  entropy $H(Y_i)=-\sum_{y\in\mathbb{Y}}\pi(y)\log \pi(y)<\infty$
  equals
  \begin{align}
    \label{UnifilarEntropy}
    h^P=
    \sum_{y\in\mathbb{Y}}\pi(y)
    \kwad{-\sum_{x\in\mathbb{X}} \varepsilon(x|y)\log \varepsilon(x|y)}
    .
  \end{align}
\end{theorem}
The entropy rate of a non-unifilar hidden Markov processs is much more
difficult to compute
\cite{Blackwell57,EphraimMerhav02,HanMarcus06,JacquetSeroussiSzpankowski08}.

Knowing that Oracle processes are unifilar, we state the following.
\begin{theorem}
  The entropy rate of the stationary Oracle($\theta$) process equals
  \begin{align}
    h^P
    &=\frac{h(\theta)+\theta}{2-\theta},
  \end{align}
  where $h(\theta):=-\theta \log\theta -(1-\theta)\log(1-\theta)$.
\end{theorem}
\begin{proof}
  Using equation (\ref{Stationary}), we can easily determine the
  stationary distribution as
  $\pi(ay)=\pi(a)\okra{\frac{\theta}{2}}^{\abs{y}}$,
  $\pi(by)=\pi(ay)(1-\theta)$, and $\pi(a)=(1-\theta)/(2-\theta)$.
  Hence
  \begin{align}
    H(Y_i)
    &=-\sum_{y\in\klam{a,b}\times\klam{0,1}^*}\pi(y)\log \pi(y)
      \nonumber\\
    &=\sum_{y\in\klam{0,1}^*}
    \kwad{-(2-\theta)\pi(ay)\log \pi(ay)-(1-\theta)\pi(ay)\log (1-\theta)}
      \nonumber\\
    &=\frac{1-\theta}{2-\theta}\sum_{k=0}^\infty
      \kwad{-(2-\theta)\theta^k k\log\frac{\theta(1-\theta)}{2(2-\theta)}
      -(1-\theta)\theta^k\log (1-\theta)}
      \nonumber\\
    &=\frac{(1-\theta)}{(2-\theta)}
      \kwad{-\frac{(2-\theta)\theta}{(1-\theta)^2}
      \log\frac{\theta(1-\theta)}{2(2-\theta)}-
      \log (1-\theta)}
      \nonumber\\
    &=\frac{
      (2-\theta)\theta
      \kwad{-\log\theta+1+\log(2-\theta)}
      -\log (1-\theta)
      }{
      (1-\theta)(2-\theta)
      }
      .
  \end{align}
  Since this entropy is finite, we can compute the entropy rate by
  Theorem \ref{theoUnifilar} as
  \begin{align}
    h^P
    &=\sum_{y\in\klam{0,1}^*}
      \pi(ay)\kwad{-\theta\log\frac{\theta}{2}-(1-\theta)\log(1-\theta)}
      \nonumber\\
    &=\pi(a)\sum_{n=0}^\infty \theta^n\kwad{h(\theta)+\theta}
      =\frac{h(\theta)+\theta}{2-\theta}.
  \end{align}
\end{proof}

Now let us proceed to the main result of this section, i.e., computing
Hilberg exponents $\beta^P_{g,z}$, $\beta^P_2$, $\beta^P_3$, and
$\beta^P_{\mathbb{M}}$ for Oracle processes and showing that they are
equal and can take arbitrary values in range $(0,1)$.  To determine
Hilberg exponent $\beta_{g,z}$, we will use quite an obvious predictor
\begin{align}
    \label{SantaFePredictorEncoded}
     g(k,x_1^n)
    :=
    \begin{cases}
      0 & \text{if $2\_\psi(k)20\sqsubseteq x_1^n$ 
        and $2\_\psi(k)21\not\sqsubseteq x_1^n$}, \\
      1 & \text{if $2\_\psi(k)21\sqsubseteq x_1^n$ 
        and $2\_\psi(k)20\not\sqsubseteq x_1^n$}, \\
      2 & \text{else}.
    \end{cases} 
\end{align}
In the above definition, symbol `$\_$' matches any symbol.
\begin{theorem}
  For predictor (\ref{SantaFePredictorEncoded}) and the stationary
  Oracle($\theta$) process, 
  \begin{align}
    \beta^P_{g,z}=\beta^P_2=\beta^P_3=\beta^P_{\mathbb{M}}
    &=\beta:=\frac{1}{1-\log\theta}.
  \end{align}
\end{theorem}
\begin{proof}
  By Theorems \ref{theoMIOrder} and \ref{theoFactsMI}, it suffices to
  show $\beta_{g,z}\ge \beta$ and $\beta^P_{\mathbb{M}}\le \beta$. The
  proof of $\beta_{g,z}\ge \beta$ will apply techniques developed in
  \cite{Debowski18} for Santa Fe processes.  The proof of
  $\beta^P_{\mathbb{M}}\le \beta$ will use some ideas from
  \cite{Debowski12} derived also for Santa Fe processes. For both
  goals of the proof, we will apply random variables $N_n\ge 0$,
  $Z_i\in\klam{0,1}$, and $W_i\in\klam{0,1}^*$ constructed through
  parsing
  \begin{align}
    X_1^{N_n}=W_02Z_0W_12Z_1W_22Z_2...W_{n}2Z_{n}.
  \end{align}
  Obviously $Z_i=z_{\phi(W_i)}$ for the Oracle($\theta$) process.  In
  contrast, by the strong Markov property, random variables
  $(W_i)_{i\in\mathbb{N}}$ form an IID process, where
  \begin{align}
    P(W_i=y)=(1-\theta)\okra{\frac{\theta}{2}}^{\abs{y}}.
  \end{align}

  Since $N_n=\sum_{i=0}^n(\abs{W_i}+2)$ and $(W_i)_{i\in\mathbb{N}}$
  is an IID process then by the Hoeff\-ding inequality
  \cite{Hoeffding63} probabilities $P(N_n<\floor{\alpha n})$ and
  $P(N_n>\floor{\alpha n})$ vanish exponentially fast for $\alpha$'s
  respectively less or greater than $\sred\abs{W_i}+2$.  We can
  evaluate for $i\ge 1$ that
  \begin{align}
    \sred\abs{W_i}=(1-\theta)\sum_{y\in\klam{0,1}^*} \abs{y}
    \okra{\frac{\theta}{2}}^{\abs{y}}=(1-\theta)\sum_{n=0}^n n
    \theta^n=\frac{\theta}{1-\theta}
    .
  \end{align}
  Hence $n\le\frac{\sred N_n}{\sred\abs{W_i}+2}\le (n+1)$.  Observe now that
  $U_g(X_1^{N_n}|z)\le n$. Thus
  \begin{align}
    \sred  U_g(X_1^{N_n}|z)\le
    \sred U_g(X_1^{\floor{\alpha n}}|z)+nP(N_n>\floor{\alpha n}).
  \end{align}
  Similarly, $\mathbb{M}(X_1^{\floor{\alpha n}})\le \floor{\alpha n}$. Thus
  \begin{align}
    \sred  \mathbb{M}(X_1^{\floor{\alpha n}})\le
    \sred \mathbb{M}(X_1^{N_n})+\floor{\alpha n}P(N_n<\floor{\alpha n}).
  \end{align}
  Hence we obtain
  \begin{align}
    \hilberg_{n\to\infty} \sred U_g(X_1^{N_n}|z)
    &\le
      \hilberg_{n\to\infty} \sred U_g(X_1^n|z)
      =\beta^P_{g,z},
    \\
    \beta^P_{\mathbb{M}}=
    \hilberg_{n\to\infty} \sred \mathbb{M}(X_1^n)
    &\le
      \hilberg_{n\to\infty} \sred \mathbb{M}(X_1^{N_n}).
  \end{align}

  Define now
  \begin{align}
    U_n&:=\min\klam{k\ge 1: \psi(k)\not\in \klam{W_i}_{i=1}^n}
         .
    \\
    M_n&:=\abs{W_0}+2+ \sum_{y\in\klam{0,1}^*} (\abs{y}+2) 
         \boole{y\in \klam{W_i}_{i=1}^n}
         .
  \end{align}
  We see that $U_g(X_1^{N_n}|z)=U_n$. Since $U_n$ is a non-decreasing
  function of $n$, we have
  \begin{align}
    \hilberg_{n\to\infty} U_n\le \hilberg_{n\to\infty} \sred U_n
    \le \beta^P_{g,z}
    \text{ $P$-almost surely}.
  \end{align}
  by Theorem A9 from \cite{Debowski18}, see also
  \cite{Debowski15d,Debowski21}. On the other hand, we see that
  $\hat{\mathbb{P}}(X_1^{N_n}|M_n)\ge P(X_1^{N_n}|Y_1)$, where $Y_1$
  is the hidden state emitting $X_1$. It is so since we can express
  probability $P(X_1^{N_n}|Y_1)$ as probability of a unifilar process
  with $M_n$ hidden states. Thus, by the Barron lemma \cite[Theorem
  3.1]{Barron85b}, we obtain
  \begin{align}
    P\okra{\mathbb{M}(X_1^{N_n})>M_n}
    &\le
      P\okra{\hat{\mathbb{P}}(X_1^{N_n}|M_n)< w_{N_n}\mathbb{P}(X_1^{N_n})}
      \nonumber\\
    &\le
      P\okra{\frac{w_{N_n}\mathbb{P}(X_1^{N_n})}{P(X_1^{N_n}|Y_1)}>1}
      \le
      \sred w_{N_n}.
  \end{align}
  In consequence, by the Hoeffding bound for $N_n$, we obtain
  \begin{align}
    \beta^P_{\mathbb{M}}\le
    \hilberg_{n\to\infty} \sred\mathbb{M}(X_1^{N_n})\le
    \hilberg_{n\to\infty} \sred M_n.
  \end{align}
  To finish the proof, it suffices to show
  \begin{align}
    \label{Goal}
    \hilberg_{n\to\infty} U_n\ge \beta\ge
    \hilberg_{n\to\infty} \sred M_n \text{ $P$-almost surely}.
  \end{align}

  To accomplish the left inequality in (\ref{Goal}), we observe
  \begin{align}
    P(U_n< 2^m)
    &\le \sum_{k=1}^{2^m-1} P(\psi(k)\not\in \klam{W_i}_{i=1}^n)
      = \sum_{y\in\klam{0,1}^{< m}} P(y\not\in \klam{W_i}_{i=1}^n)
      \nonumber\\
    &= \sum_{k=0}^{m-1} 2^k
      \okra{1-(1-\theta)\okra{\frac{\theta}{2}}^k}^n
      \le 2^m \okra{1-(1-\theta)\okra{\frac{\theta}{2}}^m}^n
     \nonumber\\
    &\le 2^m\exp\okra{-(1-\theta)n\okra{\frac{\theta}{2}}^m}
      = 2^m\exp\okra{-(1-\theta)n2^{-m/\beta}}.
  \end{align}
  Putting $m_n=\beta(1-\epsilon)\log n$ for an arbitrary $\epsilon>0$,
  we obtain
  \begin{align}
    \sum_{n=1}^\infty P(U_n< 2^{m_n})\le \sum_{n=1}^\infty
    n^{\beta(1-\epsilon)}\exp(-(1-\theta)n^\epsilon)<\infty.
  \end{align}
  Hence by the Borel-Cantelli lemma
  \begin{align}
    \beta\le \hilberg_{n\to\infty} U_n  \text{ $P$-almost surely}.
  \end{align}
  
  To accomplish the right inequality in (\ref{Goal}), we notice
  \begin{align}
    \sred M_n-\frac{2-\theta}{1-\theta}
    &= \sum_{y\in\klam{0,1}^*} (\abs{y}+2) P(y\in \klam{W_i}_{i=1}^n) 
      \nonumber\\
    &= \sum_{k=0}^{\infty} (k+2) 2^k
      \okra{1-\okra{1-(1-\theta)\okra{\frac{\theta}{2}}^k}^n}
      \nonumber\\
    &\le \sum_{k=0}^{\infty} (k+2) 2^k
      \okra{1-\okra{1-2^{-k/\beta}}^n}.
  \end{align}
  Hence, adapting the computations from the proof of Proposition 1 by
  \cite{Debowski12}, we obtain up to a small constant 
  \begin{align}
    \sred M_n
    &
      \lesssim \int_{0}^\infty
      (k+2) 2^k \okra{1-\okra{1-2^{-k/\beta}}^n} dk
      \nonumber\\
    &
      =\frac{1}{\ln 2}\int_{1}^\infty
      (\log p+2) \okra{1-\okra{1-p^{-1/\beta}}^n}
      dp \quad \klam{p:=2^k}
      \nonumber\\
    &
      =\frac{\beta^2}{\ln 2}
      \int_0^1
      \frac{(1-u)(\log(1-u^{1/n})^{-1}+2)du}{u^{1-1/n}n(1-u^{1/n})^{\beta+1}}
      \quad
      \klam{u:=\okra{1-p^{-1/\beta}}^n}
       \nonumber\\
    &
      =\frac{\beta^2n^\beta(\log n+2)}{\ln 2} \int_0^1 f_n(u) du
      + \frac{\beta^2n^\beta}{\ln 2} \int_0^1 g_n(u) du,
  \end{align}
  where we denote functions
  \begin{align}
    \label{fnu}
    f_n(u):=\frac{(1-u)}{u^{1-1/n}[n(1-u^{1/n})]^{\beta+1}}
    ,
    \quad
    g_n(u):=f_n(u)\log[n(1-u^{1/n})]^{-1}
    .
  \end{align}
  These functions tend to limits
  \begin{align}
    \lim_{n\rightarrow\infty} f_n(u) =
    f(u):=\frac{(1-u)}{u(-\ln u)^{\beta+1}},
    \quad
    \lim_{n\rightarrow\infty} g_n(u) =
    g(u):=f(u)\log(-\ln u)^{-1}.
  \end{align}
  We notice upper bounds $f_n(u)\le f(u)$ and $g_n(u)\le g_1(u)$ for
  $u\in(0,1)$.  Moreover, functions $f(u)$ and $g_1(u)$ are integrable
  on $u\in(0,1)$. Indeed putting $t:=-\ln u$ and integrating by parts
  yields
  \begin{align}
    &\int_0^1  f(u) du
      =
      \int_0^\infty (1-e^{-t})\,t^{-\beta-1} dt
      \nonumber\\
    &
      \quad
      =(1-e^{-t})(-\beta^{-1})t^{-\beta}|_0^\infty 
      +\int_0^\infty e^{-t}\beta^{-1}t^{-\beta}dt
      =\beta^{-1}\Gamma(1-\beta),
  \end{align}
  whereas putting $t=1-u$ and integrating by parts yields
  \begin{align}
    &\int_0^1  g_1(u) du
      =
      -\int_0^1 \frac{\log t}{t^\beta}dt
      \nonumber\\
    &
      \quad
      =-(\log t)(1-\beta)^{-1}t^{1-\beta}|_0^1
      +\int_0^1 (1-\beta)^{-1}t^{-\beta}dt
      =(1-\beta)^{-2}.
   \end{align}
  Hence we derive
  \begin{align}
    \hilberg_{n\to\infty} \sred M_n\le \beta. 
  \end{align}
  This completes the proof.
\end{proof}

As we can see by the above theorem, Oracle processes can have
arbitrary large Hilberg exponents
$\beta^P_{g,z}=\beta^P_2=\beta^P_3=\beta^P_{\mathbb{M}}\in(0,1)$. In
particular, the hidden Markov order estimator can diverge as a power
law even for so simple unifilar processes and it diverges at the
slowest possible rate prescribed by the bounds in Theorems
\ref{theoMIOrder} and \ref{theoFactsMI}. That is, these bounds can be
non-trivially tight.

\section*{Acknowledgment}

We thank Jan Mielniczuk, Damian Niwi\'nski, Elisabeth Gassiat, and
Peter Gr\"unwald for inspiring feedback. This work was supported by
the National Science Centre Poland grant no.\ 2018/31/B/HS1/04018.

\bibliographystyle{IEEEtran}

\bibliography{0-journals-abbrv,0-publishers-abbrv,ai,mine,tcs,ql,books,nlp}

\begin{thebibliography}{10}
\providecommand{\url}[1]{#1}
\csname url@rmstyle\endcsname
\providecommand{\newblock}{\relax}
\providecommand{\bibinfo}[2]{#2}
\providecommand\BIBentrySTDinterwordspacing{\spaceskip=0pt\relax}
\providecommand\BIBentryALTinterwordstretchfactor{4}
\providecommand\BIBentryALTinterwordspacing{\spaceskip=\fontdimen2\font plus
\BIBentryALTinterwordstretchfactor\fontdimen3\font minus
  \fontdimen4\font\relax}
\providecommand\BIBforeignlanguage[2]{{%
\expandafter\ifx\csname l@#1\endcsname\relax
\typeout{** WARNING: IEEEtran.bst: No hyphenation pattern has been}%
\typeout{** loaded for the language `#1'. Using the pattern for}%
\typeout{** the default language instead.}%
\else
\language=\csname l@#1\endcsname
\fi
#2}}

\bibitem{Gacs74en2}
P.~G\'acs, ``On the symmetry of algorithmic information,'' \emph{Sov.\ Math.\
  Dokl.}, vol.~15, pp. 1477--1480, 1974.

\bibitem{Chaitin75}
G.~J. Chaitin, ``A theory of program size formally identical to information
  theory,'' \emph{J.\ ACM}, vol.~22, pp. 329--340, 1975.

\bibitem{LiVitanyi08}
M.~Li and P.~M.~B. Vit\'anyi, \emph{An Introduction to {Kolmogorov} Complexity
  and Its Applications, 3rd ed.}\hskip 1em plus 0.5em minus 0.4em\relax
  Springer, 2008.

\bibitem{Hilberg90}
W.~Hilberg, ``{Der bekannte Grenzwert der redundanzfreien Information in Texten
  --- eine Fehlinterpretation der Shannonschen Experimente?}'' \emph{Frequenz},
  vol.~44, pp. 243--248, 1990.

\bibitem{EbelingNicolis91}
W.~Ebeling and G.~Nicolis, ``Entropy of symbolic sequences: the role of
  correlations,'' \emph{Europhys.\ Lett.}, vol.~14, pp. 191--196, 1991.

\bibitem{EbelingPoschel94}
W.~Ebeling and T.~P\"oschel, ``Entropy and long-range correlations in literary
  {English},'' \emph{Europhys.\ Lett.}, vol.~26, pp. 241--246, 1994.

\bibitem{BialekNemenmanTishby01b}
W.~Bialek, I.~Nemenman, and N.~Tishby, ``Complexity through nonextensivity,''
  \emph{Physica A}, vol. 302, pp. 89--99, 2001.

\bibitem{CrutchfieldFeldman03}
J.~P. Crutchfield and D.~P. Feldman, ``Regularities unseen, randomness
  observed: {The} entropy convergence hierarchy,'' \emph{Chaos}, vol.~15, pp.
  25--54, 2003.

\bibitem{RadfordOther19}
A.~Radford, J.~Wu, R.~Child, D.~Luan, D.~Amodei, and I.~Sutskever, ``Language
  models are unsupervised multitask learners,'' 2019,
  \url{https://openai.com/blog/better-language-models/}.

\bibitem{BrownOthers20}
T.~B. Brown, B.~Mann, N.~Ryder, M.~Subbiah, J.~Kaplan, P.~Dhariwal,
  A.~Neelakantan, P.~Shyam, G.~Sastry, A.~Askell, S.~Agarwal, G.~K. Ariel
  Herbert-Voss, T.~Henighan, R.~Child, A.~Ramesh, D.~M. Ziegler, J.~Wu,
  C.~Winter, C.~Hesse, M.~Chen, M.~L. Eric~Sigler, S.~Gray, B.~Chess, J.~Clark,
  C.~Berner, S.~McCandlish, A.~Radford, I.~Sutskever, and D.~Amodei, ``Language
  models are few-shot learners,'' 2020, \url{https://arxiv.org/abs/2005.14165}.

\bibitem{TakahiraOthers16}
R.~Takahira, K.~Tanaka-Ishii, and {\L}.~D\k{e}bowski, ``Entropy rate estimates
  for natural language---a new extrapolation of compressed large-scale
  corpora,'' \emph{Entropy}, vol.~18, no.~10, p. 364, 2016.

\bibitem{HahnFutrell19}
M.~Hahn and R.~Futrell, ``Estimating predictive rate-distortion curves via
  neural variational inference,'' \emph{Entropy}, vol.~21, p. 640, 2019.

\bibitem{BravermanOther19}
M.~Braverman, X.~Chen, S.~M. Kakade, K.~Narasimhan, C.~Zhang, and Y.~Zhang,
  ``Calibration, entropy rates, and memory in language models,'' 2019,
  \url{https://arxiv.org/abs/1906.05664}.

\bibitem{KaplanOther20}
J.~Kaplan, S.~McCandlish, T.~Henighan, T.~B. Brown, B.~Chess, R.~Child,
  S.~Gray, A.~Radford, J.~Wu, and D.~Amodei, ``Scaling laws for neural language
  models,'' 2020, \url{https://arxiv.org/abs/2001.08361}.

\bibitem{Debowski06}
{\L}.~D\k{e}bowski, ``On {Hilberg}'s law and its links with {Guiraud}'s law,''
  \emph{J.\ Quantit.\ Linguist.}, vol.~13, pp. 81--109, 2006.

\bibitem{Debowski11b}
------, ``On the vocabulary of grammar-based codes and the logical consistency
  of texts,'' \emph{IEEE Trans.\ Inform.\ Theory}, vol.~57, pp. 4589--4599,
  2011.

\bibitem{Debowski18}
------, ``Is natural language a perigraphic process? {The} theorem about facts
  and words revisited,'' \emph{Entropy}, vol.~20, no.~2, p.~85, 2018.

\bibitem{Debowski21}
------, \emph{Information Theory Meets Power Laws: Stochastic Processes and
  Language Models}.\hskip 1em plus 0.5em minus 0.4em\relax Wiley, 2021, in
  press.

\bibitem{DeMarcken96}
C.~G. de~Marcken, ``Unsupervised language acquisition,'' Ph.D. dissertation,
  Massachussetts Institute of Technology, 1996.

\bibitem{KiefferYang00}
J.~C. Kieffer and E.~Yang, ``Grammar-based codes: {A} new class of universal
  lossless source codes,'' \emph{IEEE Trans.\ Inform.\ Theory}, vol.~46, pp.
  737--754, 2000.

\bibitem{CharikarOthers05}
M.~Charikar, E.~Lehman, A.~Lehman, D.~Liu, R.~Panigrahy, M.~Prabhakaran,
  A.~Sahai, and A.~Shelat, ``The smallest grammar problem,'' \emph{IEEE Trans.\
  Inform.\ Theory}, vol.~51, pp. 2554--2576, 2005.

\bibitem{KuraszkiewiczLukaszewicz51en}
W.~Kuraszkiewicz and J.~{\L}ukaszewicz, ``The number of different words as a
  function of text length,'' \emph{Pami\k{e}tnik Literacki}, vol. 42(1), pp.
  168--182, 1951, in Polish.

\bibitem{Guiraud54}
P.~Guiraud, \emph{{Les caract{\`e}res statistiques du vocabulaire}}.\hskip 1em
  plus 0.5em minus 0.4em\relax Paris: Presses Universitaires de France, 1954.

\bibitem{Herdan64}
G.~Herdan, \emph{Quantitative Linguistics}.\hskip 1em plus 0.5em minus
  0.4em\relax Butterworths, 1964.

\bibitem{Heaps78}
H.~S. Heaps, \emph{Information Retrieval---Computational and Theoretical
  Aspects}.\hskip 1em plus 0.5em minus 0.4em\relax Academic Press, 1978.

\bibitem{MerhavGutmanZiv89}
N.~Merhav, M.~Gutman, and J.~Ziv, ``On the estimation of the order of a
  {Markov} chain and universal data compression,'' \emph{IEEE Trans.\ Inform.\
  Theory}, vol.~35, no.~5, pp. 1014--1019, 1989.

\bibitem{ZivMerhav92}
J.~Ziv and N.~Merhav, ``Estimating the number of states of a finite-state
  source,'' \emph{IEEE Trans.\ Inform.\ Theory}, vol.~38, no.~1, pp. 61--65,
  1992.

\bibitem{CsiszarShields00}
I.~Csiszar and P.~C. Shields, ``The consistency of the {BIC Markov} order
  estimator,'' \emph{Ann.\ Statist.}, vol.~28, pp. 1601--1619, 2000.

\bibitem{Csiszar02}
I.~Csiszar, ``Large-scale typicality of {Markov} sample paths and consistency
  of {MDL} order estimator,'' \emph{IEEE Trans.\ Inform.\ Theory}, vol.~48,
  no.~6, pp. 1616--1628, 2002.

\bibitem{MorvaiWeiss05}
G.~Morvai and B.~Weiss, ``Order estimation of {Markov} chains,'' \emph{IEEE
  Trans.\ Inform.\ Theory}, vol.~51, no.~4, pp. 1496--1497, 2005.

\bibitem{PeresShields05}
Y.~Peres and P.~Shields, ``Two new {Markov} order estimators,'' 2005,
  \url{https://arxiv.org/abs/math/0506080}.

\bibitem{DaleviDubhashi05}
D.~Dalevi and D.~Dubhashi, ``The {Peres-Shields} order estimator for fixed and
  variable length {Markov} models with applications to {DNA} sequence
  similarity,'' in \emph{Algorithms in Bioinformatics}, R.~Casadio and
  G.~Myers, Eds.\hskip 1em plus 0.5em minus 0.4em\relax Springer, 2005, pp.
  291--302.

\bibitem{RyabkoAstola06}
B.~Ryabko and J.~Astola, ``Universal codes as a basis for time series
  testing,'' \emph{Statist.\ Methodol.}, vol.~3, pp. 375--397, 2006.

\bibitem{CsiszarTalata06}
I.~Csiszar and Z.~Talata, ``Context tree estimation for not necessarily finite
  memory processes, via {BIC} and {MDL},'' \emph{IEEE Trans.\ Inform.\ Theory},
  vol.~52, pp. 1007--1016, 2006.

\bibitem{Talata13}
Z.~Talata, ``Divergence rates of {Markov} order estimators and their
  application to statistical estimation of stationary ergodic processes,''
  \emph{Bernoulli}, vol.~19, no.~3, pp. 846--885, 2013.

\bibitem{BaigorriGocalvesResende14}
A.~R. Baigorri, C.~R. Goncalves, and P.~A.~A. Resende, ``Markov chain order
  estimation based on the chi-square divergence,'' \emph{Canad.\ J.\ Statist.},
  vol.~42, no.~4, pp. 563--578, 2014.

\bibitem{RyabkoAstolaMalyutov16}
B.~Ryabko, J.~Astola, and M.~Malyutov, \emph{Compression-Based Methods of
  Statistical Analysis and Prediction of Time Series}.\hskip 1em plus 0.5em
  minus 0.4em\relax Springer, 2016.

\bibitem{PapapetrouKugiumtzis16}
M.~Papapetrou and D.~Kugiumtzis, ``Markov chain order estimation with
  parametric significance tests of conditional mutual information,''
  \emph{Sim.\ Model.\ Pract.\ Theory}, vol.~61, pp. 1--13, 2016.

\bibitem{Finesso90}
L.~Finesso, ``Order estimation for functions of {Markov} chains,'' Ph.D.
  dissertation, University of Maryland, 1990.

\bibitem{WeinbergerLempelZiv92}
M.~J. Weinberger, A.~Lempel, and J.~Ziv, ``A sequential algorithm for the
  universal coding of finite memory sources,'' \emph{IEEE Trans.\ Inform.\
  Theory}, vol.~38, no.~3, pp. 1002--1014, 1992.

\bibitem{Kieffer93}
J.~C. Kieffer, ``Strongly consistent code-based identification and order
  estimation for constrained finite-state model classes,'' \emph{IEEE Trans.\
  Inform.\ Theory}, vol.~39, no.~3, pp. 893--902, 1993.

\bibitem{WeinbergerFeder94}
M.~J. Weinberger and M.~Feder, ``Predictive stochastic complexity and model
  estimation for finite-state processes,'' \emph{J.\ Statist.\ Plan.\ Infer.},
  vol.~39, pp. 353--372, 1994.

\bibitem{LiuNarayan94}
C.-C. Liu and P.~Narayan, ``Order estimation and sequential universal data
  compression of a hidden {Markov} source bv the method of mixtures,''
  \emph{IEEE Trans.\ Inform.\ Theory}, vol.~40, no.~4, pp. 1167--1180, 1994.

\bibitem{GassiatBoucheron03}
E.~Gassiat and S.~Boucheron, ``Optimal error exponents in hidden markov models
  order estimation,'' \emph{IEEE Trans.\ Inform.\ Theory}, vol.~49, no.~4, pp.
  964--980, 2003.

\bibitem{Lehericy19}
L.~Leh\'ericy, ``Consistent order estimation for nonparametric {Hidden Markov
  Models},'' \emph{Bernoulli}, vol.~25, no.~1, pp. 464--498, 2019.

\bibitem{ShaliziShaliziCrutchfield03}
C.~R. Shalizi, K.~L. Shalizi, and J.~P. Crutchfield, ``An algorithm for pattern
  discovery in time series,'' 2003, \url{http://www.arxiv.org/abs/cs/0210025}.

\bibitem{ZhengHuangTong19}
J.~Zheng and J.~H. aand Changqing~Tong, ``The order estimation for hidden
  markov models,'' \emph{Physica A}, vol. 527, p. 121462, 2019.

\bibitem{GrayKieffer80}
R.~M. Gray and J.~C. Kieffer, ``Asymptotically mean stationary measures,''
  \emph{Ann.\ Probab.}, vol.~8, pp. 962--973, 1980.

\bibitem{FontanaGrayKieffer81}
R.~Fontana, R.~Gray, and J.~Kieffer, ``Asymptotically mean stationary
  channels,'' \emph{IEEE Trans.\ Inform.\ Theory}, vol.~27, pp. 308--316, 1981.

\bibitem{KiefferRahe81}
J.~C. Kieffer and M.~Rahe, ``Markov channels are asymptotically mean
  stationary,'' \emph{SIAM J.\ Math.\ Anal.}, vol.~12, no.~3, pp. 293--305,
  1981.

\bibitem{Debowski09}
{\L}.~D\k{e}bowski, ``A general definition of conditional information and its
  application to ergodic decomposition,'' \emph{Statist.\ Probab.\ Lett.},
  vol.~79, pp. 1260--1268, 2009.

\bibitem{Debowski12}
------, ``Mixing, ergodic, and nonergodic processes with rapidly growing
  information between blocks,'' \emph{IEEE Trans.\ Inform.\ Theory}, vol.~58,
  pp. 3392--3401, 2012.

\bibitem{Skinner57}
B.~F. Skinner, \emph{Verbal Behavior}.\hskip 1em plus 0.5em minus 0.4em\relax
  Prentice Hall, 1957.

\bibitem{Chomsky56}
N.~Chomsky, ``Three models for the description of language,'' \emph{IRE Trans.\
  Inform.\ Theory}, vol.~2, no.~3, pp. 113--124, 1956.

\bibitem{Chomsky57}
------, \emph{Syntactic Structures}.\hskip 1em plus 0.5em minus 0.4em\relax The
  Hague: Mouton \& Co, 1957.

\bibitem{Chomsky59}
------, ``A review of {B. F. Skinner's Verbal Behavior},'' \emph{Language},
  vol.~35, no.~1, pp. 26--58, 1959.

\bibitem{ChomskyMiller59}
N.~Chomsky and G.~Miller, ``Finite state languages,'' \emph{Inform.\ Control},
  vol.~1, pp. 91--112, 1959.

\bibitem{Ryabko88en2}
B.~Y. Ryabko, ``Prediction of random sequences and universal coding,''
  \emph{Probl.\ Inform.\ Transm.}, vol.~24, no.~2, pp. 87--96, 1988.

\bibitem{Ryabko08}
B.~Ryabko, ``Compression-based methods for nonparametric density estimation,
  on-line prediction, regression and classification for time series,'' in
  \emph{2008 IEEE Information Theory Workshop, Porto}, 2008, pp. 271--275.

\bibitem{Shtarkov87en2}
Y.~M. Shtarkov, ``Universal sequential coding of single messages,''
  \emph{Probl.\ Inform.\ Transm.}, vol. 23(2), pp. 3--17, 1987.

\bibitem{Grunwald07}
P.~D. Gr\"unwald, \emph{The Minimum Description Length Principle}.\hskip 1em
  plus 0.5em minus 0.4em\relax The MIT Press, 2007.

\bibitem{Debowski20c}
{\L}.~D\k{e}bowski, ``On a class of {Markov} order estimators based on {PPM}
  and other universal codes,'' 2020, \url{https://arxiv.org/abs/2003.04754}.

\bibitem{Barron85b}
A.~R. Barron, ``Logically smooth density estimation,'' Ph.D. dissertation,
  Stanford University, 1985.

\bibitem{ZivLempel77}
J.~Ziv and A.~Lempel, ``A universal algorithm for sequential data
  compression,'' \emph{IEEE Trans.\ Inform.\ Theory}, vol.~23, pp. 337--343,
  1977.

\bibitem{Wyner78}
A.~D. Wyner, ``A definition of conditional mutual information for arbitrary
  ensembles,'' \emph{Inform.\ Control}, vol.~38, pp. 51--59, 1978.

\bibitem{Debowski20}
{\L}.~D\k{e}bowski, ``Approximating information measures for fields,''
  \emph{Entropy}, vol.~22, no.~1, p.~79, 2020.

\bibitem{Debowski15d}
------, ``Hilberg exponents: New measures of long memory in the process,''
  \emph{IEEE Trans.\ Inform.\ Theory}, vol.~61, pp. 5716--5726, 2015.

\bibitem{Brudno82}
A.~A. Brudno, ``Entropy and the complexity of trajectories of a dynamical
  system,'' \emph{Trans.\ Mosc.\ Math.\ Soc.}, vol.~44, pp. 124--149, 1982.

\bibitem{GacsKorner73}
P.~G\'acs and J.~K\"{o}rner, ``Common information is far less than mutual
  information,'' \emph{Probl.\ Contr.\ Inform.\ Theory}, vol.~2, pp. 119--162,
  1973.

\bibitem{Wyner75}
A.~D. Wyner, ``The common information of two dependent random variables,''
  \emph{IEEE Trans.\ Inform.\ Theory}, vol. IT-21, pp. 163--179, 1975.

\bibitem{Debowski17}
{\L}.~D\k{e}bowski, ``Regular {Hilberg} processes: An example of processes with
  a vanishing entropy rate,'' \emph{IEEE Trans.\ Inform.\ Theory}, vol.~63,
  no.~10, pp. 6538--6546, 2017.

\bibitem{Debowski10}
------, ``Variable-length coding of two-sided asymptotically mean stationary
  measures,'' \emph{J.\ Theor.\ Probab.}, vol.~23, pp. 237--256, 2010.

\bibitem{TraversCrutchfield11}
N.~F. Travers and J.~P. Crutchfield, ``Exact synchronization for finite-state
  sources,'' \emph{J.\ Statist.\ Phys.}, vol. 145, pp. 1181--1201, 2011.

\bibitem{TraversCrutchfield11b}
------, ``Asymptotic synchronization for finite-state sources,'' \emph{J.\
  Statist.\ Phys.}, 2011.

\bibitem{TraversCrutchfield14}
------, ``Infinite excess entropy processes with countable-state generators,''
  \emph{Entropy}, vol.~16, pp. 1396--1413, 2014.

\bibitem{Blackwell57}
D.~Blackwell, ``The entropy of functions of finite-state {Markov} chains,'' in
  \emph{Transactions of the First {Prague} Conference on Information Theory,
  Statistical Decision Functions, Random Processes}.\hskip 1em plus 0.5em minus
  0.4em\relax Czechoslovak Academy of Sciences, 1957, pp. 13--20.

\bibitem{EphraimMerhav02}
Y.~Ephraim and N.~Merhav, ``Hidden {Markov} processes,'' \emph{IEEE Trans.\
  Inform.\ Theory}, vol.~48, pp. 1518--1569, 2002.

\bibitem{HanMarcus06}
G.~Han and B.~Marcus, ``Analyticity of entropy rate of hidden {Markov} chain,''
  \emph{IEEE Trans.\ Inform.\ Theory}, vol.~52, pp. 5251--5266, 2006.

\bibitem{JacquetSeroussiSzpankowski08}
P.~Jacquet, G.~Seroussi, and W.~Szpankowski, ``On the entropy of a hidden
  {Markov} process,'' \emph{Theor.\ Comput.\ Sci.}, vol. 395, no. 2--3, pp.
  203--219, 2008.

\bibitem{Hoeffding63}
W.~Hoeffding, ``Probability inequalities for sums of bounded random
  variables,'' \emph{J.\ Amer.\ Statist.\ Association}, vol.~58, no. 301, pp.
  13--30, 1963.

\end{thebibliography}

\end{document}